\documentclass{article}
\usepackage[utf8x]{inputenc}%pacchetti
\usepackage[english]{babel}
\usepackage{xcolor}
\usepackage{float}
\usepackage[all,cmtip]{xy} %arrows
\usepackage{wrapfig}
\usepackage[all]{xy}
\usepackage{relsize}
\usepackage{bbm}
\usepackage{symate}
\usepackage{enumitem}
\usepackage[english]{varioref}
\usepackage[nowrite]{frontespizio}
\renewcommand{\theta}{\vartheta}
\renewcommand{\phi}{\varphi}
\newtheorem{theorem}{Theorem}[section]
\newtheorem{proposition}{Proposition}[section]

\newtheorem{rem}{Remark}[section]
\newtheorem{definition}{Definition}[section]
\def\cA{\mathcal{A}}
\def\cambi#1{{#1}}

\title{Nekhoroshev theorem for perturbations of the central motion}
\author{D. Bambusi\footnote{Dipartimento di Matematica, Universit\`a degli Studi di Milano, Via Saldini 50, I-20133
Milano. \newline
 \textit{Email: } \texttt{dario.bambusi@unimi.it}},
 A. Fus\`e\footnote{Dipartimento di Matematica, Universit\`a degli Studi di Milano, Via Saldini 50, I-20133
Milano. \newline
 \textit{Email: } \texttt{alessandra.fuse@unimi.it}}}

\begin{document}

\maketitle

\begin{abstract}
In this paper we prove a Nekhoroshev type theorem for perturbations of Hamiltonians describing a particle subject to the force due to a central potential. Precisely, we prove that under an explicit condition on the potential, the Hamiltonian of the central motion is quasi-convex. Thus, when it is perturbed, two actions (the modulus of the total angular momentum and the action of the reduced radial system) are approximately conserved for times which are exponentially long with the inverse of the perturbation parameter.
\end{abstract}

\section{Introduction}\label{intro}

In this paper we study the applicability of Nekhoroshev's theorem
\cite{Nek1,Nek2} (see also
\cite{BGG85,lochak1992canonical,GM96,N04,N06,guzzo2016steep}) to the
central motion. The main point is that Nekhoroshev's theorem applies
to perturbations of integrable systems whose Hamiltonian is a
\emph{steep} function of the actions. Even if such a property is known
to be generic, it is very difficult (and not at all explicit) to
verify it. Here we prove that, under an explicit condition on the
potential (see eq. \eqref{PotCondition}), the Hamiltonian of the
central motion is a quasi-convex function of the actions and thus it
is steep, so that Nekhoroshev's theorem applies. Actually, the form of
Nekhoroshev's theorem used here is not the original one, but that
for degenerate systems proved by Fassò in
\cite{fasso1995hamiltonian}. This is due to the fact that the
Hamiltonian of the central motion is a function of two actions only,
namely, the modulus of the total angular momentum and the action of
the effective one dimensional Hamiltonian describing the motion of the
radial variable.

On the one hand, as pointed out in \cite{fasso1995hamiltonian}, this
feature creates some problems for the proof of Nekhoroshev's theorem,
but these problems were solved in \cite{fasso1995hamiltonian}. On the
other hand, degeneracy reduces the difficulty for the verification of
steepness or of the strongest property of being quasi-convex, since,
in the two-dimensional case, quasi-convexity is generic and equivalent
to the nonvanishing of the Arnold determinant

\begin{equation}
\label{arndet}
\mathcal{D}:=\det\left(\begin{aligned}
&\frac{\partial^2 h_0}{\partial I^2} & \left(\frac{\partial h_0}{\partial I}\right)^t\\
&\frac{\partial h_0}{\partial I}  & 0
\end{aligned}\right)\ ,
\end{equation}
a property that it is not too hard to verify.

\noindent Indeed, since \eqref{arndet} is an analytic function of the actions,
it is enough to show that it is different from zero at one point in
order to ensure that it is different from zero almost everywhere. Here
we explicitly compute the expansion of $h_0(I)$ at a circular orbit
and we show that, provided the central potential $V(r)$ does not
satisfy identically a fourth order differential equation that we
explicitly write, the Hamiltonian $h_0(I)$ is quasi-convex on an open
dense domain (whose complementary is possibly empty).

The rest of the paper is organized as follows: In
sect. \ref{statement} we introduce the central motion problem and
state the main results. Sect. \ref{proof} contains all the proofs. In
the Appendix we prove that in the two dimensional case quasi-convexity
is equivalent to Arnold isoenergetic nondegeneracy condition.

\noindent{\it Acknowledgements.} We thank F. Fass\`o for a detailed
discussion on action angle variables in the central motion problem,
M. Guzzo, L. Niederman and G. Pinzari for pointing to our attention
some relevant references and A. Maspero for interesting discussions.

\section{Preliminaries and statement of the main results}\label{statement}

We first recall the structure of the action angle variables for the
central motion. Introducing polar coordinates, the Hamiltonian takes
the form
\begin{equation}\label{HamCentral}
h_0(p_r,r,p_\phi,\phi,p_\theta,\theta)=\frac{p_r^2}{2}+\frac{p_\theta^2}{2r^2}+\frac{p_\phi^2}{2r^2\sin^2\theta}+V(r) \;,
\end{equation}
and the actions on which $h_0$ depends are
$$I_2:=\sqrt{p_\theta^2+\frac{p_\phi^2}{\sin^2\theta}} \; ,$$
and the action $I_1$ of the effective one dimensional Hamiltonian system
\begin{equation}\label{HamReduced}
h_0^*=\frac{p_r^2}{2}+V_{I_2}^*(r) \; , \quad
V_{I_2}^*(r)=\frac{I_2^2}{2r^2}+V(r) \; .
\end{equation}
 By construction $h_0$ turns out to be a function of the two actions only. We still write
$$h_0=h_0(I_1,I_2) \; . $$ 

According to Fassò's theorem, if $h_0$
 depends on $(I_1,I_2)$ in a steep way, then Nekhoroshev's estimate
 applies. We recall that steepness is actually implied by
 quasi-convexity, the property that we are now going to verify.

\begin{definition}
A function $h_0$ of the actions is said to be \emph{quasi-convex} at a
point $\bar{I}$ if the system
$$\left\lbrace\begin{aligned}
& d h_0(\bar{I})\eta =0\\
&d^2h_0(\bar{I})(\eta,\eta)=0
\end{aligned}\right.$$
admits only trivial solutions. Here we denoted by $d^2h_0(\bar{I})(\eta,\eta)$ the second differential of $h_0$ at $\bar{I}$ applied to the two vectors $\eta,\eta$.
\end{definition}

To define the set $\cA$ in which the actions vary we first assume that
there exists an interval $(r_2,r_1)$ such that, for $r_2<r<r_1$ one has
\begin{align}
\label{nd.1}
V'(r) > 0\ ,
\\
\label{nd.2}
V''(r) + \frac{3 V'(r)}{r} > 0\ .
\end{align}
Then we define 
$$
\Gamma_1:=\sqrt{r_1^3 V'(r_1)}\ ,\quad \Gamma_2:=\sqrt{r_2^3
  V'(r_2)}\ ,
$$ and, in order to fix ideas, we assume that
$\Gamma_1<\Gamma_2$. Then for $\Gamma_1<I_2<\Gamma_2$, the effective
potential $V_{I_2}^*$ has a non degenerate minimum at some
$r_0=r_0(I_2)$.

Then, there exists a curve $E(I_2)$ such that for $h_0^* < E(I_2)$,
all the orbits of the Hamiltonian \eqref{HamReduced} are
periodic. Correspondingly, their action $I_1(E,I_2)$ vary in some
interval $(0, F(I_2))$. Thus, the domain $\mathcal{A}$ of the actions
$I$ has the form
$$\mathcal{A}:=\{ (I_1,I_2): \Gamma_1 < I_2 < \Gamma_2 , \; 0< I_1 <
F(I_2)\} \; .$$ We remark that $\mathcal{A}$ is simply connected, a
property that will play an important role in the following.

Our main result is the following.

\begin{theorem}\label{Teorema1}
Consider the Hamiltonian
$$h_0(p_r,r,p_\phi,\phi,p_\theta,\theta)=\frac{p_r^2}{2}+\frac{p_\theta^2}{2r^2}+\frac{p_\phi^2}{2r^2\sin^2\theta}+V(r)
\; ,$$ with $V(r)$ analytic on $\mathbb{R}^3\smallsetminus
\{0\}$. Assume that there exists a value $r_0\in (r_2,r_1)$ of the
radius such that the following fourth order equation
\begin{equation}
\label{PotCondition}
\begin{aligned}
V^{(4)}(r_0)=&-\frac{84 V'(r_0)}{r_0^3}+\frac{32 V''(r_0)}{r_0^2}+\frac{16V'''(r_0)}{r_0}-\frac{8(V''(r_0))^2}{r_0V'(r_0)}\\
&+\frac{240(V'(r_0))^2}{r_0^3(3V'(r_0)+r_0V''(r_0))}-\frac{40V'(r_0)V'''(r_0)}{r_0(3V'(r_0)+r_0V''(r_0))}\\
&+\frac{5r_0(V'''(r_0))^2}{3(3V'(r_0)+r_0V''(r_0))}
\end{aligned}
\end{equation}
is \emph{not} satisfied.

Then, there exists a subset $\mathcal{S}\subset \mathcal{A}$ of the
action space, with the property that its intersection with any compact
set is composed by at most a finite number of lines, and such that $h_0$
restricted to $\mathcal{A}\smallsetminus \mathcal{S}$ is
quasi-convex.
\end{theorem}

\begin{rem}
The fourth order equation \eqref{PotCondition} can be rewritten as a
second order ordinary differential equation in terms of the variable
$g(r)=\frac{r V''(r)}{V'(r)}$, namely,
\begin{equation}
\label{SecondOrd}
\begin{aligned}
g''(r_0)=&\frac{(14+g(r_0))g'(r_0)}{3r_0}+\frac{(g(r_0)-1)(g(r_0)+2)(g(r_0)+3)}{3r_0^2}\\
&+\frac{5g'(r_0)^2}{3(3+g(r_0))} 
\end{aligned}
\end{equation}
\end{rem}

\begin{rem}
It is interesting to see what are the homogeneous potentials which do
not give rise to steep Hamiltonians. Thus take 
 $$V(r)=k r^\alpha\; ,$$
 with $\alpha, k\in\mathbb{R}$, then the assumptions of Theorem
 \ref{Teorema1} are fulfilled if
\begin{equation}
\label{nd.3}
k\alpha>0\ ,\ \quad \alpha+2>0\ ,\quad \alpha\not=-1,2\ ,
\end{equation}
thus the excluded cases are the Kepler and the Harmonic potentials.

We also remark that the equation \eqref{PotCondition} has also the
solution $\alpha=-2$, which however is excluded by the second of
\eqref{nd.3}. 
\end{rem}

\cambi{The theory of \cite{nekhoroshev1972action} and
  \cite{fasso1995hamiltonian} is needed in order to apply Fass\`o's
  version of Nekhoroshev's theorem. We recall that the theory of
  \cite{nekhoroshev1972action,fasso1995hamiltonian} applyes provided
  $\cA$ is simply connected and $\forall a\in\cA$, the set $I^{-1}(a)$
  is compact. This second property follows from the remark that in our
  case
\begin{equation}
\label{livello}
I^{-1}(a)=\left\{(p_r,r,p_\phi,\phi,p_\theta,\theta) :
h_0^*(p_r,r,p_\phi,\phi,p_\theta,\theta)\leq h_0(a_1,a_2)
 \right\}\ ,
\end{equation}
where $h_0^*$ is constructed using $I_2:=a_2$. Then the set
 \eqref{livello} is obviously compact. It follows from
 \cite{nekhoroshev1972action,fasso1995hamiltonian} that the
 set
$$ {I}^{-1}(\mathcal{A})
:=\{(p_r,r,p_\phi,\phi,p_\theta,\theta)\ :\ {I}(p_r,r,p_\phi,\phi,p_\theta,\theta)
\in \mathcal{A} \}$$ can be covered by charts defining generalized
action angle coordinates of the form $(I,\alpha,x,y)$ with $\alpha
\in\mathbb{T}^2$.
}

\begin{definition}
Fix a positive parameter $\rho$ and denote by $B_\rho(I)\subset
\mathbb{R}^2$ the open ball of radius $\rho$ and center $I$ and define
\begin{equation}
\label{Set}
\mathcal{S}_\rho :=\cup_{I\in \mathcal{S}} B_\rho(I)\ .
\end{equation}
\end{definition}

We now consider a small perturbation $\varepsilon f$ of $h_0$, with
$f$ a function of the original cartesian coordinates $(p,q)$ in $T
\mathbb{R}^3 \simeq \mathbb{R}^6$ which is analytic.

\begin{theorem}\label{Teorema2}
Fix a positive small parameter $\rho$ and consider the Hamiltonian $h:=h_0 + \varepsilon f$.
Then, for every compact set $\mathcal{K}\subset \mathcal{A}$, there exist positive constants $\varepsilon^*,\mathcal{C}_1,\mathcal{C}_2,\mathcal{C}_3$ such that if the initial value $I_0$ of the actions fulfills $I_0\in \mathcal{K}\smallsetminus \mathcal{S}_\rho$ and $\lvert \varepsilon\rvert < \varepsilon^* $ one has
$$\lvert I(t)-I_0\rvert \leq \mathcal{C}_1\varepsilon^{1/4} \; ,$$
for all times $t$ satisfying
\begin{equation}\label{times}
\lvert t \rvert \leq \mathcal{C}_3 exp(\mathcal{C}_2\varepsilon^{-1/4}) \; .
\end{equation}
\end{theorem}

\begin{rem}
The main dynamical consequence is that, as in the central motion, for any initial datum as above, there exist $r_m,r_M$ such that 
$$r_m \leq r \leq r_M$$
for the times \eqref{times}.
\end{rem}

\begin{rem}
Actually, one can weaker the analyticity requirement on $f$, since it
would be enough to have that it is analytic on the set of $p$ and
$q$'s such that the action is close to $I_0$.
\end{rem}

\begin{rem}
The Theorem holds also if one couples the system to an other system with a dynamic taking place over a much faster time scale. For example, this occurs in the case of a soliton interacting with radiation in the NLS equation as in \cite{bambusi2016freezing}.
\end{rem}

\section{Proof of Theorem \ref{Teorema1}}\label{proof}

As anticipated above, in the case of \emph{two} actions,
quasi-convexity is equivalent to the nonvanishing of the Arnold
determinant $\mathcal{D}$ (cf. \eqref{arndet}).

In order to compute $\mathcal{D}$, we have to compute
$h_0(I_1,I_2)$. To do this we proceed as follows.

First, as explained in sect. \ref{statement}, it is easy to introduce
the action $I_2$ which coincides with the modulus of the total angular
momentum. Then, $I_1$ is the action of the effective one dimensional
system \eqref{HamReduced} in which $I_2$ plays the role of a
parameter. In order to have an explicit formula for the Hamiltonian as
a function of the actions, we work at circular orbits. Precisely, we
exploit the remark that in one dimensional systems Birkhoff normal
form converges in a neighborhood of a nondegenerate minimum. Indeed, Birkhoff
normal form allows to construct an analytic canonical transformation
which, in a neighborhood of the critical point, conjugates the
Hamiltonian $h_0$ to a function of the form
\begin{equation}
\label{hi}
h_0\left(\frac{p_r^2+r^2}{2}, I_2 \right) \; ,
\end{equation}
which moreover is explicitly constructed as a power series in
$\frac{p_r^2+r^2}{2}$. Thus, one can define the first action $I_1$ by
$I_1:=\frac{p_r^2+r^2}{2} $. Remark that, since \eqref{hi}, as a
function of $(p_r,r)$ is analytic in a whole \emph{complex}
neighborhood of zero, then $h_0(I_1,I_2)$ is analytic for $I_1$ in a
whole neighborhood of zero. Then, from uniqueness of the actions in
one dimensional systems, one has that the expression one gets is
actually the expression of $h_0$ as a function of the actions as
defined in sect. \ref{statement}. It also follows that
$\mathcal{D}(I_1,I_2)$ as a function of $I_1$ extends to a complex
analytic function in whole neighborhood of zero, and such a function
can be computed using the expression of $h_0$ obtained from
\eqref{hi}.

We use this remark in order to compute the Arnold determinant at the
equilibrium point $r_0$ of the effective one dimensional system
described by $h_0$, which coincides with a circular orbit of the
original system. This can be done by computing explicitly the second
order Taylor expansion of $h_0$ which coincides with the fourth order
Birkhoff normal form of $h_0^*$ at $r_0$. Actually this was already
done in \cite{fejoz2004theoreme} getting
\begin{equation}\label{HamExp}
h(I_1,I_2)=V^*(I_2)+\sqrt{A(I_2)}I_1+\frac{-5B(I_2)^2+3C(I_2)A(I_2)}{48A(I_2)^2}I_1^2+o(I_1^2)\; ,
\end{equation}
where
$$\begin{aligned}
&V^*(I_2)=\frac{I_2^2}{2r_0^2}+V(r_0) 	\; , &A(I_2)=\frac{3 I_2^2}{r_0^4}+V''(r_0) \; ,\\
&B(I_2)=-\frac{12 I_2^2}{r_0^5}+V'''(r_0) \; , & C(I_2)=\frac{60 I_2^2}{r_0^6}+V^{(4)}(r_0)\; .
\end{aligned}$$

Fix a point $I_2^*$ and let $r_0(I_2^*)$ be the corresponding critical
point of the effective potential.

Inserting into the Arnold determinant $\mathcal{D}$, the first and second derivatives of the Hamiltonian \eqref{HamExp} evaluated at the point $I^*:=(0,I_2^*)$,  one gets
$$\mathcal{D}=\det\left( \begin{matrix}
t(I_2^*)   &     \frac{1}{2\sqrt{A(I_2^*)}}\CD{\frac{6I_2^*}{r_0^4}+B(I_2^*)\frac{\partial r_0}{\partial I_2}}   &      &\sqrt{A(I_2^*)}\\
 \frac{1}{2\sqrt{A(I_2^*)}}\CD{\frac{6I_2^*}{r_0^4}+B(I_2^*)\frac{\partial r_0}{\partial I_2}}     &   \frac{1}{r_0^2}-\frac{2(I_2^*)^2}{r_0^3}\cdot\frac{\partial r_0}{\partial I_2}            &      &\frac{I_2^*}{r_0^2}\\
\sqrt{A(I_2^*)}                                                        &           \frac{I_2^*}{r_0^2}                                               &     &                      0
\end{matrix}\right) \; ,$$
where we denoted by 
$$t(I_2^*)=\frac{-5B(I_2^*)^2+3C(I_2^*)A(I_2^*)}{24A(I_2^*)^2} \; .$$
Thus, $\mathcal{D}=0$ is equivalent to
$$\frac{6(I_2^*)^2}{r_0^6}+\CD{\frac{BI_2^*}{r_0^2}+\frac{2AI_2^*}{r_0^3}}\frac{\partial r_0}{\partial I_2}-\frac{A}{r_0^2}-\frac{t(I_2^*)^2}{r_0^4}=0 \; .$$
Isolating the fourth order derivative of $V$ one gets
\begin{equation}\label{Pot}
V^{(4)}(r_0)=\frac{48A}{r_0^2}+\CD{\frac{8ABr_0^2}{I_2^*}+\frac{16A^2r_0}{I_2^*}}\cdot \frac{\partial r_0}{\partial I_2}-\frac{8A^2r_0^2}{(I_2^*)^2}+\frac{5B^2}{3A}-\frac{60(I_2^*)^2}{r_0^6} \; .
\end{equation}
Using that $r_0$ is a critical point of $V^*$, one can express $I_2^*$ in terms of $r_0$, namely,
$$I_2(r_0)=\sqrt{r_0^3V'(r_0)} \; ,$$
and computing
$$\frac{\partial r_0}{\partial I_2}=\frac{1}{\frac{\partial I_2}{\partial r_0}}=\frac{2(r_0^3V'(r_0))^{1/2}}{r_0^2\CD{3V'(r_0)+r_0V''(r_0)}} \; ,$$
one can rewrite the equation \eqref{Pot} in terms of the radius $r_0$, obtaining the equation \eqref{PotCondition}. Thus, if there exists $r_0$ such that \eqref{PotCondition} is not satisfied, then at this point $\mathcal{D}\neq 0$. This concludes the proof of the Theorem.
\begin{flushright}
$\square$
\end{flushright}

\noindent\textbf{Proof of Theorem \ref{Teorema2}}
From the previous result, one gets that the Hamiltonian is quasi-convex on the subset $\mathcal{K}\smallsetminus \mathcal{S}_\rho$ with uniform bounds on the quasi-convexity constants. Thus, if we choose an initial datum in such a set and a sufficiently small parameter $\varepsilon$, then Fassò's version of the Nekhoroshev theorem for degenerate Hamiltonians applies (see \cite{fasso1995hamiltonian} for details).

Furthermore, we remark that, in the case of quasi-convex Hamiltonians, one gets optimal {Nekhoroshev's exponents} which for two dimensional systems are $1/4$.
\begin{flushright}
$\square$
\end{flushright}

\section{Appendix}

In this appendix we show that, in the two dimensional case, quasi-convexity is equivalent to the nonvanishing of the Arnold determinant $\mathcal{D}$. We start by recalling the Arnold condition.
\begin{definition}
Let $h_0$ be a complete integrable Hamiltonian with $n$ degrees of freedom and frequency $\omega$. Then, $h_0$ is said to satisfy the \emph{Arnold condition} at $I^*$ if the following map
$$(I,\lambda)\rightarrow (\lambda\omega(I),h_0(I))$$
has maximal rank at $(I^*,1)$.
\end{definition}
Explicitly, this condition can be written in the form
$$\mathcal{D}(I^*)=\det \left (  
\begin{aligned}
&\frac{\partial \omega(I^*)}{\partial I}  &  \frac{\partial h_0(I^*)}{\partial I}\\
&\omega(I^*)                                           &  0
\end{aligned} \right) \neq 0 \; .$$

\begin{proposition}
Let $h_0: \mathcal{A} \rightarrow \mathbb{R}$ with $\mathcal{A}\subset \mathbb{R}^2$ be an Hamiltonian. Then, $h_0$ is quasi-convex at $I^*\in\mathcal{A}$ if and only if $\mathcal{D}(I^*)\neq 0$.
\end{proposition}

\begin{proof}
In the two dimensional case, $\mathcal{D}\neq 0$ takes the form
$$\omega_1\CD{\frac{\partial^2 h_0}{\partial I_1 \partial I_2}\omega_2 - \frac{\partial^2 h_0}{\partial I_2^2}\omega_1}-\omega_2\CD{\frac{\partial^2 h_0}{\partial I_1^2}\omega_2-\frac{\partial^2h_0}{\partial I_1\partial I_2}\omega_1}\neq 0 \; ,$$
namely,
$$\frac{\partial^2 h_0}{\partial I_1^2}\omega_2^2-2\frac{\partial^2 h_0}{\partial I_1\partial I_2}\omega_1\omega_2+\frac{\partial^2 h_0}{\partial I_2^2}\omega_1^2\neq 0$$ 
where all the quantities are evaluated at the point $I^*$.

Moreover, this condition can be explicitly written as 
$$d^2h_0(I^*)(\eta,\eta) \neq 0 \; ,$$
where we denoted by $\eta=(\omega_2, -\omega_1)$.

Thus, we conclude that, in the case $n=2$, the Arnold condition is equivalent to the request of the second differential $d^2h_0(I^*)$ to be different from zero on the hyperplane generated by the vector $\eta$ normal to the gradient, namely, quasi-convexity.
\end{proof}

%\bibliography{Biblio} 
%\bibliographystyle{alpha}

\end{document}